\documentclass{amsart}

\usepackage[english]{babel}
\usepackage{bbm}
\usepackage{graphicx}
\usepackage{url}
\usepackage{mathtools}
\usepackage{verbatim}
\usepackage{bigints}

\graphicspath{{Fig/}}

\newcommand{\diff}{\mathop{}\mathopen{}\mathrm{d}}
\newcommand{\ind}[1]{\ensuremath{\mathbbm{1}_{\{#1\}}}}

\newcommand{\N}{\ensuremath{\mathbb{N}}}
\newcommand{\R}{\ensuremath{\mathbb{R}}}

\newcommand{\E}{\ensuremath{\mathbb{E}}}
\renewcommand{\P}{\ensuremath{\mathbb{P}}}
\newcommand\croc[1]{\left\langle #1\right\rangle}
\newtheorem{proposition}{Proposition}
\newtheorem{lemma}{Lemma}
\newtheorem{theorem}{Theorem}

\newcommand{\eps}{\varepsilon}
\newcommand{\cal}{\mathcal}

\title[Stochastic Model of Replication]{ Analysis of a Stochastic Model of Replication\\ in Large Distributed Storage Systems:\\A Mean-Field Approach}

\author[W. Sun]{Wen Sun}\thanks{The author's work is supported by a public grant overseen by the French National Research Agency (ANR) as part of the ``Investissements d'Avenir'' program (reference: ANR-10-LABX-0098).}
\email{Wen.Sun@inria.fr}

\author[V. Simon]{V\'eronique Simon}
\email{vsimon19@gmail.com }

\author[S. Monnet]{S\'ebastien Monnet}
\email{Sebastien.Monnet@univ-smb.fr} 
\address{LISTIC - Polytech Annecy-Chamb\'ery, BP 80439,  74944 Annecy le Vieux, France}

\author[Ph. Robert]{Philippe  Robert}
\email{Philippe.Robert@inria.fr}
\urladdr{http://team.inria.fr/rap/robert}
\address[Ph. Robert,W. Sun]{INRIA Paris, 2 rue Simone Iff, F-75012 Paris, France}

\author[P. Sens]{Pierre Sens}
\email{Pierre.Sens@lip6.fr} 
\address[V. Simon,P. Sens]{Sorbonne Universit\'es, UPMC Univ Paris 06, CNRS, INRIA, LIP6 UMR 7606,\\ 4 place Jussieu 75005 Paris.}

\begin{document}

\maketitle

\begin{abstract}
  Distributed storage systems such as Hadoop  File System or Google File System (GFS) ensure data availability and durability using replication. Persistence is achieved by replicating the same data block on several nodes, and ensuring that a minimum number of copies are available on the system at any time. Whenever the contents of a node are lost, for instance due to a hard disk crash, the system regenerates the data blocks stored before the failure by transferring them from the remaining replicas.  This paper is focused on the analysis of the efficiency of replication mechanism  that determines  the location of the copies of a given file at some server. The variability of the loads of the nodes of the network is investigated for several policies. Three  replication mechanisms are tested against simulations in the context of a  real implementation of a such a system: Random, Least Loaded and Power of Choice.

  The simulations show that some of these policies  may lead to quite unbalanced situations: if $\beta$ is the average number of copies per node it turns out that, at equilibrium, the load of  the nodes may exhibit a high variability. It is shown in this paper that a simple variant of a power of choice type algorithm has a striking effect on the loads of the nodes: at equilibrium, the distribution of the load of a node has a bounded support,  most of nodes have a load less than $2\beta$ which is an interesting property for the design of the storage space of these systems.

Mathematical models are introduced and investigated  to explain this  interesting phenomenon.  The analysis of these systems turns out to be quite complicated mainly because of the large dimensionality of the state spaces involved. Our study relies on probabilistic methods, mean-field analysis,  to analyze the asymptotic behavior of an arbitrary node of the network when the total number of nodes gets large. An additional ingredient is the use of stochastic calculus with marked Poisson point processes to establish some of our results. 

\end{abstract}

\section{Introduction}
For scalability, performance or for fault-tolerance concerns in distributed storage systems, the pieces of data are spread among many distributed nodes. Most famous distributed data stores include  Google File System (GFS)~\cite{Ghemawat03}, Hadoop Distributed File System (HDFS)~\cite{borthakur2008hdfs}, Cassandra~\cite{Lakshman:2010:CDS:1773912.1773922}, Dynamo~\cite{DeCandia:2007:DAH:1294261.1294281}, Bigtable~\cite{Chang:2006:BDS:1267308.1267323}, PAST~\cite{past} or DHASH~\cite{dhash}.  Most systems rely on data redistribution. Large amounts of data have to be stored in a distributed and reliable manner. They use a hash function in the case of distributed hash tables (DHTs)~\cite{past, dhash}.  As shown in previous studies, these systems imply many data movements and may lose data under churn~\cite{Legtchenko:2012:RCR:2240166.2240178}. Rodrigues and Blake have shown that classical DHTs storing large amounts are usable only if the node lifetime is of the order of several days~\cite{IPTPSRodriguesB04}.

Distributed data storage permits to enhance access performance by spreading the load among many nodes. It can also improve fault tolerance by maintaining multiple copies of each piece of data.  While implementing a distributed data store, many problems have to be tackled. For instance, it is necessary to efficiently locate a given piece of data: to balance the storage load evenly among nodes, to maintain consistency and the fault-tolerance level.  While consistency and fault-tolerance in replicated data stores are widely studied, the storage load balance received little attention despite its importance. The distribution of the storage load among the storing nodes is a critical issue. On a daily basis, new pieces of data have to be stored and when a failure occurs,  maintenance mechanisms are supposed to create and store new copies to replace the lost ones. A key feature of these systems is that the storage infrastructure itself is dynamic: nodes may crash and new nodes may be added. If the placement policy used does not balance the storage load evenly among nodes, the imbalance will become harmful. The overloaded nodes may have to serve many more requests than the other nodes, and in case of failure, the recovery procedure will take more time, increasing the probability to lose data.

Although it is not mentioned explicitly in the description of most of these systems,  the design of some parts of  DHT's  is reminiscent of   peer-to-peer systems architectures. But these are not the only framework where DHTs can be used.  One of the best examples of such a system is Cassandra~\cite{Lakshman}. It is a {\em fully centralized } DHT, with failure detection mechanisms comparable to the ones considered in this paper. See  the {\tt failure-detection} section of the corresponding web site	{\tt http://cassandra.apache.org/}.   It has been initially developed by Facebook and is now used by  companies such as GitHub, Instagram, Netflix, Reddit, eBay\ldots  The  placement strategies investigated in this paper  can therefore be used in various architectures, not only for peer-to-peer distributed hash-tables.

In this paper we study data placement policies avoiding data redistribution: once a piece of data is assigned to a node, it will remain on it until the node crashes. We focus specifically on the evaluation of the impact of several placement strategies on the storage load balance on a long term. To the best of our knowledge,  there are few papers devoted to the analysis of the evolution of the  storage load of the nodes of a DHT system  on such a long term period. Our investigation has been done in two complementary steps. 

\medskip
\noindent
{\em A simulation environment} of a real system  based on PeerSim~\cite{zzz_peersim}  is used to emulate several years of evolution of this system for three placement policies which are defined below: Random, Least Loaded and Power of Choice. See Figures~\ref{fig:load_evol_199} and~\ref{fig:load_distrib}.

\medskip
\noindent
    {\em
        Simplified mathematical models} are presented to analyze the Random and Power of Choice Policies. Mean-field results are obtained when the number $N$ of nodes gets large. It should be stressed that a number of aspects are not taken into account in the mathematical models: delays to copy files, network congestion due to duplication or losses of files, \ldots  See Section~\ref{SecMath} for the motivation and more details.

    We also consider only the steady state of these systems in our results, mainly for the sake of mathematical tractability. %Nevertheless, there are heuristic arguments that suggest that these systems converge to equilibrium quite rapidly in fact, even with a low failure rate. We take the case of the random policy analyzed in Section~\ref{SecMath}. First, for classical mean-field results, and it can be shown that this is the case here, the sequence of empirical distributions converges to the mean-field at a rate of at least of the order of $1/\sqrt N$, uniformly on finite time intervals.  Secondly, Theorem~\ref{prop1}, the asymptotic process itself is a simple birth and death process which converges also very rapidly to equilibrium.
    These mathematical models appear nevertheless to explain some of the phenomena  concerning the load of the nodes observed in the simulations.

 The Least Loaded policy is, without a surprise, quite optimal, the load  of the nodes being almost constant in this case,  it varies only within some small set of values.  We show that, for a large network with an average load $\beta$ per node ,  if $\overline{X}_\beta^R$, [resp. $\overline{X}_\beta^P$] is the load of a random node at equilibrium for the Random policy [resp. Power of choice policy] then, for $x\geq 0$, 
\begin{align}\label{J1}
  \lim_{\beta\to+\infty} \P\left(\frac{\overline{X}_\beta^R}{\beta}\geq x\right)&=e^{-x},\notag\\
  \lim_{\beta\to+\infty} \P\left(\frac{\overline{X}_\beta^P}{\beta}\geq x\right)&=
  \begin{cases}
    1{-}{x}/{2}& \text{ if } x<2,\\
    0 & \text{ if } x\geq 2.
  \end{cases}
\end{align}
See Theorems~\ref{th1} and~\ref{ThPow} below.

The striking feature is that, for the Power of choice policy, the distribution of the load of a node has a {\em finite support} $[0,2\beta]$ for a large average load per node $\beta$. This is an important and desirable property  for the design of such systems, to dimension the storage of the nodes in particular. Note that this is not the case for the Random policy. Our simulations of a real system exhibit this surprising phenomenon, even for moderately large loads, see Figure~\ref{fig:load_distrib}. Another interesting  feature is the fact that, in the limit, the distribution of the load of a node is {\em uniform} on $[0,2\beta]$. It should be noted that the finite support feature is only an asymptotic property, for large $N$ and $\beta$, of the distribution of the load of a node. Additionally it does not imply, of course, that the maximum of the loads of the nodes is bounded.

Usually Power of choice policies used in computer science and communication networks are associated with $\log\log N$ loads instead of $\log N$ loads, see Mitzenmacher~\cite{Mitzenmacher}; or with double exponential decay for the tail distribution of the load at equilibrium, instead of an exponential decay, see Vvedenskaya et al.~\cite{Dobrushin}. Here the phenomenon is that  the number of files stored at a node is bounded in the limit, i.e. it has  a finite support,  instead of an exponential decay for the tail distribution of this variable.

The  mathematical analysis of these systems turns out to be quite complicated mainly because of the large dimensionality of the state spaces involved. Our study relies on probabilistic methods to analyze the asymptotic behavior of an arbitrary node of the network when the total number of nodes gets large. An additional ingredient is the use of stochastic calculus with marked Poisson point processes to establish some of our results.

The paper is organized as follows.  The main placement policies are introduced in Section~\ref{SecDef}. Section~\ref{SecSim} describes the simulation model and presents the results obtained with the simulator. Concerning mathematical models, the Random policy is analyzed in Section~\ref{SecRand} and Power of Choice policy in Section~\ref{SecPC}. All (quite) technical details of the proofs of the results for the Random policy are included. This is not the case for the Power of choice policy, for sake of simplicity and due to the much more complex framework of general mean-field results,  convergence results of the sample paths (Proposition~\ref{thP}) and of the invariant distributions (Proposition~\ref{ConvInvPow}) are stated without proof. A reference is provided. The complete proofs of the important convergence results~\eqref{J1} are provided.

\section{Placement policies}\label{SecDef}
To each data block is associated a \emph{root node}, a node having a copy of the block in charge of its duplication if necessary.  During the recovery process to replace a lost copy, the root node has to choose a new storage node within a dedicated set of nodes, {\em the selection range} of the node. Any node of this subset that does not already store a copy of the same data block may be chosen. The mechanism in charge of the failure of the root nodes is beyond the scope of this paper and the selection range is assumed to be the set of all nodes. Three policies of placement are defined below when the root node of a data block has to copy it on another node. 

\subsection*{Least Loaded Policy} For this algorithm the root node  selects the least loaded node of its selection range not already storing a copy of the same data block. This strategy clearly aims at reducing the variation of storage loads among nodes.
As it has been seen in earlier studies, this policy has a bad impact on the system reliability, see~\cite{splad}. Indeed, a node having a small storage load will be chosen by all its neighbors in the ring. Furthermore, this policy implies for a root node to monitor the load of all  nodes, which may be costly. It is nevertheless in terms of placement an optimal policy. It is used in this paper as a reference for comparison with the other policies. 

\subsection*{Random Policy} The root node chooses uniformly at random a new storage node  among nodes not already hosting a copy of the same data block.

\subsection*{Power of Choice Policy}
For this algorithm,  the root node chooses, uniformly at random, two nodes not storing a copy of the  data block.  It selects the least loaded among the two.

It is inspired by algorithms studied by Mitzenmacher and others in the context of static allocation schemes of balls into bins in computer science, see~\cite{Mitzenmacher} for a survey. In queueing theory, a similar algorithm has been investigated in  the seminal work of  Vvedenskaya et al.~\cite{Dobrushin} in 1996. There is a huge literature on these algorithms in this context. Our framework is quite different, the placement is dynamic, data blocks have to move because of crashes, and the number of files is constant in the system contrary to  open queueing models. The idea is nevertheless the same: reducing the load by just checking some finite subset of nodes instead of all of them. In fact the common version of this algorithm consists in taking $k$ nodes at random and choosing the least loaded node, this is the power of $k$ choices algorithm. For simplicity, we have chosen to refer to the algorithm as ``power of choice'' instead of the more accurate ``power of two choices''.

Essentially,  Random is the policy used for the two main classes of DHT architectures: Past and Chord. It does not use any information on the states of the nodes and has therefore a low overhead from this point of view. Using more detailed information may prove to be useful but will involve more messages between nodes and, therefore, will have a  cost in terms of overhead.  The least loaded policy, for example, has a high overhead since a node has to know the states of all nodes to allocate copies. This is why we compare this "optimal policy"  with a policy like power of choice which has a limited overhead but interesting performances. 

\section{Simulations}\label{SecSim}
Our simulator is based on PeerSim~\cite{zzz_peersim}, see also~\cite{p2p09-peersim}. It simulates a real distributed storage system. Every node, every piece of data, and every transfer is represented. Each piece of data is replicated and each copy is assigned to a different storage node. We describe briefly the failure detection mechanism used. 
In  classical systems, like Microsoft FeePastry/PAST implementation on a distributed infrastructure, see~\cite{past},  the routing layer frequently exchanges many messages. Thus, on each node, the neighbor lists  are updated very frequently. At storage level, the neighbor lists can be consulted to check the presence of the neighbors, and thus to detect node failures. The duration of time  between consecutive  maintenance checkings is an order of magnitude longer than the checkings on the routing layer. It is the way PAST detects nodes that join or leave the storage system in practice.

  In our simulator, for performance purposes, we did not simulate each message exchange at the routing layer level. When a node fails it is labeled as "failed" and its neighbors will consider it as failed at their next periodical maintenance. The maintenance at node ${\cal N}$  consists in 
  \begin{itemize}
  \item[(i)] checking the presence of all nodes storing data blocks for which ${\cal N}$  is the root. In the case of  faults, node ${\cal N}$  starts then creating, for each lost data-block, a new copy using a remaining one (and selecting a new storage node according to the chosen strategy).
  \item[(ii)] Checking the presence of all nodes being root for data blocks stored by ${\cal N}$. In the case of   faults, node ${\cal N}$  computes the identity of the new root for this data block. It sends a message to the new root node to notify it of its new role. The information of this change of root node for this data block is also sent to the nodes having a copy of it. 
  \end{itemize}
  The detailed algorithms and  description of the associated meta-data can be found in~\cite{Legtchenko:2012:RCR:2240166.2240178}.
  
\paragraph{System model}
\noindent We have simulated $N$ nodes, storing $F_N^*$ data blocks with a fixed size $s$ and replicated $d$ times. The nodes and the data blocks are assigned unique identifiers ($id$). The nodes are organized according to their identifiers, forming a virtual ring, as it is usual in distributed hash tables (DHTs)~\cite{past,dhash}. To each data block is associated a \emph{root node}, a node having a copy of the block in charge of its duplication if necessary. See below.

\paragraph{ Failure model}
\noindent  Failures in the systems are assumed to occur according to a Poisson process with a fixed mean of seven days.  The failures are crashes: a node acts correctly until it fails. After a crash it stops and never comes back again (fail-stop model). All the copies stored become unavailable at that time.  To maintain the number of nodes constant equal to $N$, each time a node fails, an empty node with a new $id$ joins the system in a random position in the ring of nodes.

The Poisson assumption to represent the successive failures of servers may not be completely accurate but given the large number of nodes and that the failures occur independently, the Poissonnian nature of the number of failures in a given time interval can be seen as a consequence of the law of small numbers (during some time interval each server fails with a small probability, independently of the other servers). See Pinheiro et al.~\cite{Goog}. The assumption that the number of nodes is constant is made for convenience  so that the average load per node remains constant. This is not the case in practice but the fluctuations are nevertheless not really significant. See~\cite{Goog} and the beginning of Section~\ref{SecMath}.

\paragraph{Simulation parameters}
In the simulations, based on PeerSim, the parameters have been fixed as follows:
\begin{itemize}
\item[---] The number of nodes $N{=}200$,
\item[---] the number of data blocks $F_N^*{=}10 000$,
\item[---] the block size $s{=}10$MB,
\item[---] the replication degree of data blocks $d{=}3$,
\item[---] the mean time between failures (MTBF) is $7$ days.
\end{itemize}
The network latency is fixed to $0.1s$ and the bandwidth is $5.5$Mbps. 

At the beginning of each simulation, the $F_N^*$ blocks and their copies are placed using the corresponding policy and the system is simulated for a period of $2$ years. We have studied the storage load distribution and its time evolution. With these parameters, the average load  is $\beta{=}d{\times}F_N^*/N{=}150$ blocks per node. The optimal placement from the point of view of load balancing would consist of having $150$ blocks at every node. We will investigate the deviation from this scenario for the three policies. 

\paragraph{Network simulation}
\noindent The impact of policies on bandwidth management has been carefully monitored. In case of failure, many data blocks have to be transferred among a subset of nodes to repair the system.  Taking into account bandwidth limitation and network contention is crucial since a delayed recovery may lead to the loss of additional blocks because of additional crashes in the meantime. 

\subsection{Simulation results}

Figure~\ref{fig:load_evol_199} shows the evolution of the average load of a node with respect to the duration of its lifetime within the network. One can conclude that:
\begin{itemize}
\item[---] For the \emph{Least Loaded} strategy, the load remains almost constant and equal to the optimal value $150$ .  By systematically choosing the least loaded node  to store a data block copy, the storage load tends to be constant among nodes.

  As observed in simulations, this policy has however  two main drawbacks. First, it requires that nodes maintain an up-to-date knowledge of the load of all the nodes.  Second, it is more likely  to generate network contention for the following reason: If one of the nodes is really underloaded, it will receive most of the transfer requests of its neighborhood. See~\cite{splad}.
  \item[---] For the {\em Random} strategy, the load increases linearly until the failure of the node.This is an undesired feature since it implies that the failure of ``old'' nodes will imply in this case a lot of transfers to recover the large number of lost blocks. 
\item[---] The growth of the {\em Power of Choice} policy is slow as it can be seen from the figure. It should be noted that, contrary to the Least Loaded Policy, the required information to allocate data blocks is  limited.  Indeed, its cost is only of a communication with each of the two  nodes chosen.  Furthermore, the random choice of  nodes for the allocation of  copies of files has the advantage of spreading the load from the point of view of network contention. 
\end{itemize}
\begin{figure}[ht]
\center
\includegraphics[width=0.7\textwidth]{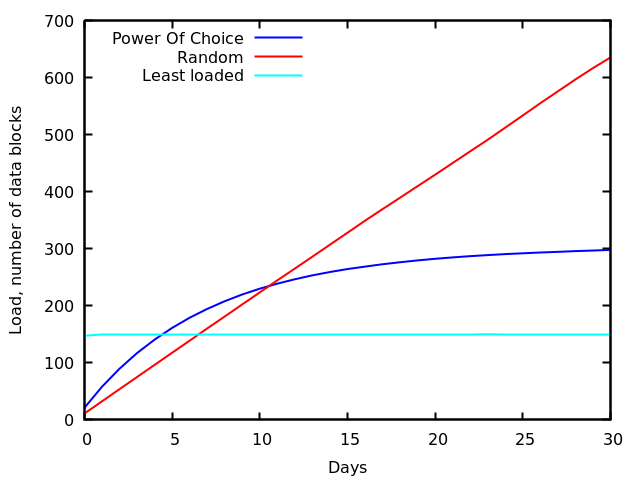}
\caption{Evolution of the Average Load of a Node with Respect to its Age in the System.}
\label{fig:load_evol_199}
\end{figure}

The distribution function  of the storage loads after two simulated years is presented in Figure~\ref{fig:load_distrib}. For clarity, the figure has been truncated.   Each point of each policy has been obtained with $210$ runs.  At the beginning,  the data block copies are placed using the corresponding strategy.  After two years of failures and reparations,  one gets that:
\begin{itemize}
\item[---] The {\em Random} strategy  presents a highly non-uniform distribution profile, note that more $10\%$ of the nodes have a loaded greater than $350$. This is consistent with our previous remark on the fact that old nodes are  overloaded.
\item[---] For the \emph{Least Loaded} strategy, as expected, the load is highly concentrated around $150$. 
\item[---] The striking feature concerning the {\em Power of Choice} policy  is that the load of a node seems to a uniform distribution between $0$ and $300$. In particular almost all nodes have a load bounded by $300$ which is absolutely remarkable. 
\end{itemize}
Table~\ref{tab:maxload} gives the maximum loads that have been observed for each strategy over $132{,}090$ samples: starting from day $100$, the maximal load has been measured and recorded every day, this for the $210$ runs.  We can see that the mean maximum load of the random strategy is already high (more than five times the average), and furthermore, the load varies a lot, the maximum measured load being $2188$ data blocks. This implies that, with the random strategy, the storage device for each node has to be  over-sized, recall that the average load is $150$ data blocks. 

\begin{figure}[ht]
\center
\includegraphics[width=0.7\textwidth]{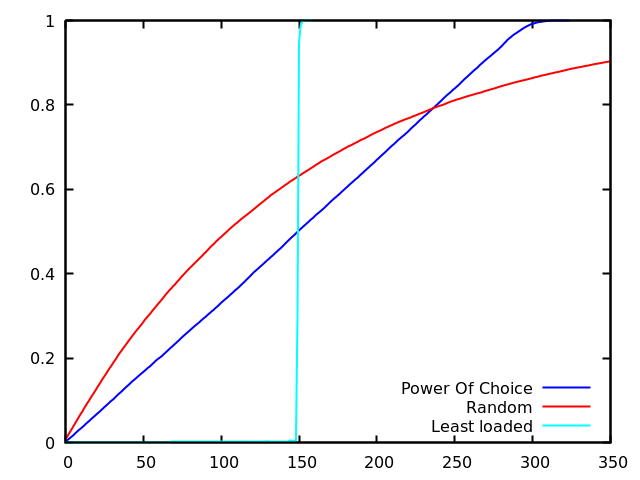}
\caption{Distribution function of load distribution of a random node after 729 days.}
\label{fig:load_distrib}
\end{figure}

\begin{table}[ht]
\begin{center}
\begin{tabular}{|l|c|c|c|}
\hline
Strategy & Mean of Max & Min & Max \\
\hline
{\em Least Loaded} & 153 & 150 & 165 \\
{\em Random} & 864 & 465 & 2188 \\
{\em Power of Choice} &  300 & 269 & 328 \\
\hline
\end{tabular}
\end{center}
\vspace{2mm}
\caption{Statistics of Maximal Loads}
\label{tab:maxload}
\end{table}

As a conclusion, the simulations show that, with a limited cost in terms of complexity,  the power of choice policy has remarkable properties. The load of each node is bounded by $300$. It may be remarked that each possible load between $0$ and $300$ is represented by the same amount of nodes on average.  Figure~\ref{fig:load_distrib} shows that there is approximately the same number of nodes having $0$ data blocks, than nodes having $150$ data block or nodes having $300$ data blocks.  Note that this is a stationary state.  Additionally, the variation is low, we can observe in Table~\ref{tab:maxload} that among the $132{,}090$ samples, the most loaded node was never above 328.  From a practical point of view, it means that a slightly oversized storage device at each node (a bit more than twice the average) is enough to guarantee the durability of the system.

In the following sections we investigate simplified mathematical models of two placement policies: Random and Power of Choice. The goal is to explain these striking qualitative properties of these policies.

\section{Mathematical Models}\label{SecMath}
The main goal of the paper is to investigate the performance of duplication algorithms in terms of the overhead for the loads of the nodes of the network. Without loss of generality, we will assume that the breakdown of each server occurs according to a Poisson process with rate $1$. After a breakdown, a server restarts empty (in fact a new server replaces it).  The replication  degree of data blocks is $d{\geq}2$, each data block has at most $d$ copies. A complete Markovian description of such a system is quite complex. Indeed,  if there are $N$ servers and $F_N^*$ initial data blocks,  for $1\leq i\leq F_N^*$, the locations of the $i$th data block are given by the indices of $k$ distinct servers if there are $k{\leq}d$ copies of this data block. Consequently the size of the state space of the Markov process is at least of the order of $(N^d)^{F_N^*}$ which is huge if it is remembered that $F_N^*$ is of the order of $N$. For this reason, we shall simplify the mathematical model. 

\paragraph{Assumption on Duplication Rate}
\noindent In order to focus specifically on the efficiency of the replacement strategy from the point of view of the distribution of the load of an arbitrary node,  we will study the system with the assumption that it does not lose files. We will only track the location of the  node of each copy of a data block with a simplifying assumption: just before  a node fails, all the copies it contains are allocated to the other nodes with respect to the algorithm of placement investigated. In this way, every data block has always $d$ copies in the system.

Note that this system is like the original one by assuming that the time to make a new copy is null. Once a server  has failed,  a copy of each of the data blocks it contains is produced immediately with one of the copies in the network. With this model, a copy could be made on the same node as the other copy, but this occurs with probability $1/(N{-}1)$, it can be shown that this has a negligible effect at the level of the network,  in the same way as in Proposition~\ref{techprop} below.  This approximation is intuitively reasonable to describe the original evolution of the system when few data blocks are lost. As we will see, qualitatively, its behavior is close to the observed simulations of the real system, few files were lost after two years. 

Now $F_N=dF_N^*$ denotes the total number of copies of files, it is assumed that, for some $\beta>0$,
\[
\lim_{N\to+\infty} \frac{F_N}{N}=\beta.
\]
$\beta$ is therefore the average load per server.  With these assumptions, the replication factor does not play a role, it is as if there were $F_N$ distinct files and once a node breaks down, any file present on this node is immediately copied to another node according to the policy used. 

If the initial state of the system is $(L_{i}^N(0))$, where $L_{i}^N(0)$ is the number of files at node $i$ initially, 
throughout the paper, it is assumed that that the distribution of the variables $(L_{i}^N(0))$  are invariant by any permutation of indices, i.e. it is an exchangeable vector, and  that
\begin{equation}\label{IniCond}
\sup_{N\geq 1} \E\left(L_1^N(0)^2\right)<+\infty. 
\end{equation}
Note that this condition is satisfied if we start with an optimal exchangeable allocation, i.e. for which, for all $1{\leq}i{\leq}N$, $$L_i^N(0)\in\{k_N{-}1, k_N\} \text{  with } k_N{=}\left\lceil {F_N}/{N}\right\rceil,$$
where, for $x{\geq} 0$, $\lceil x\rceil{\in}\N$ and ${\lceil} x\rceil{-}1{\leq} x{<}\lceil x\rceil$. 

\subsection{The Random Allocation}\label{SecRand}
For $1\leq i\leq N$, we denote by $\overline{{\cal N}}_{i}=(t_n^i,\overline{U}^{i}_n)$ the marked Poisson point process defined as follows:

\noindent
\begin{itemize}
\item[---] $(t_n^i)$ is a Poisson process on $\R_+$ with rate $1$;
\item[---] $\overline{U}^{i}_n{=}(U^{i,n}_p)$ is an i.i.d. sequence of uniform random variables on the subset $\{1,\ldots,N\}{\setminus}\{i\}$.
\end{itemize}
For $1\leq i\leq N$ and $n\geq 1$,  $t_n^i$ is the instant of the $n$th breakdown of server $i$.  For $p\geq 1$,  $U^{i,n}_p$ is the server where the $p$th copy present on node $i$ is allocated after this breakdown.  The random variables $\overline{{\cal N}}_{i}$, $1{\leq}i{\leq} N$ are assumed to be independent.  Concerning marked Poisson point processes, see Kingman~\cite{Kingman} for example. 

One will use an integral representation for these processes,  if ${\cal M}_U{=}\{1,\ldots,N\}^\N$ and $f:\R_+{\times}{\cal M}_U\to\R_+$,
\[
\sum_{n\geq 1} f(t_n^i,(U^{i,n}_p))=\int_{t=0}^{+\infty}\int_{\overline{u}=(u_p)\in{\cal M}_U} f(t,\overline{u})\overline{\cal N}_{i}(\diff t,\diff \overline{u}).
\]

\paragraph{Equations of Evolution} 
\noindent
For $1{\leq}i{\leq}N$ and $t{\geq}0$, $L_i^N(t)$ is the number of copies on server $i$ at time $t$.  
The dynamics of the random allocation algorithm is represented by the following stochastic differential equation, for $1{\leq} i{\leq} N$,
\begin{multline}\label{eqr}
\diff L_i^N(t)\stackrel{\text{def.}}{ =} L_i^N(t){-}L_i^N(t{-})= {-}L_i^N(t{-}) \overline{{\cal N}}_{i}(\diff t,{\cal M}_U )\\ +\sum_{m\not=i} \int_{{\cal M}_U}z_i\left(L_m^N(t{-}),\overline{u}\right) \overline{{\cal N}}_{m}(\diff t,\diff \overline{u})
\end{multline}
where $z_i:\N{\times}{\cal M}_U\mapsto \N$ is the function 
\begin{equation}
z_{i}(\ell, \overline{u})=\ind{u_1=i}+\ind{u_2=i}+\cdots+\ind{u_\ell=i}.
\end{equation}
The first term of the right hand side of Relation~\eqref{eqr} corresponds to a breakdown of node~$i$, all files are removed from the node. The second concerns the files added to node $i$ when  other servers break down and send  copies to node $i$. Note that the $i$th term of the sum is always $0$.

Denote $$L^N(t){=}\left(L_i^N(t),1\leq i\leq N\right)\in \N^N,$$ then clearly $(L^N(t))$ is a Markov process. Note that, because of the symmetry of the initial state and of the dynamics of the system, the variables $L_i^N(t)$ have the same distribution and since the sum of these variables is $F_N$, one has in particular $$\E\left(L_i^N(t)\right)={F_N}/{N},$$ for all $N\geq 2$, $1{\leq} i{\leq} N$ and $t\geq 0$.

The integrand in the second term of the right hand side of Equation~\eqref{eqr} has a binomial distribution with parameters $L_m^N(t)$ and $1/(N{-}1)$ and the sum of these terms is $F_N/(N{-}1)$ which is converging to $\beta$. Hence, this suggests, via an extension of the law of small numbers, that this second term could be a Poisson process with rate $\beta$. The process $(L_1^N(t))$ should be in the limit, a jump process with a Poissonnian input and return to $0$ at rate $1$.  This is what we are going to prove now. 

By integrating Equation~\eqref{eqr} one gets the relation
\begin{equation}\label{eq3}
L_1^N(t)=L_1^N(0)-\int_0^tL_1^N(s)\,\diff s+\frac{1}{N{-}1}\sum_{m=2}^N\int_0^t L_m^N(s)\,\diff s+M_1^N(t),
\end{equation}
where $(M_1^N(t))$ is the martingale
\begin{multline*}
M_1^N(t)=
-\bigintsss_0^t L_1^N(s{-})\left[\rule{0mm}{4mm} \overline{{\cal N}}_{1}(\diff s,{\cal M}_U ){-}\diff s\right]\\ +\sum_{m=2}^N \bigintsss_0^t \left[\int_{{\cal M}_U} z_1\left(L_m^N(s{-}),\overline{u}\right) \overline{{\cal N}}_{m}(\diff s,\diff \overline{u}){-}\frac{ L_m^N(s)}{N{-}1}\diff s\right].
\end{multline*}

The following proposition shows that the process $(L_1^N(t))$  does not have jumps of size $\geq 2$ on a finite time interval with high probability.
\begin{proposition}\label{Q=L}
  For $T>0$ then 
  \[
  \lim_{N\to+\infty} \P\left( \sup_{t\in(0,T]} \diff L_1^N(t)\geq 2 \right)=0.
  \]
\end{proposition}
\begin{proof}
  For $0{<}\eps{<}1$, from Equation~\eqref{eq2} in the Appendix, one obtains that there exists some constant  $K{>}0$ such that $\P({\cal E}_N){\geq}1{-}\eps$  holds for all $N{\geq}2$,  if
  \[
    {\cal E}_N=\left\{\sup\left( L_1^N(t): 0\leq t\leq T\right)\leq K\right\}.
    \]
On the event ${\cal E}_N$ the probability that a failure of some node will send more than $2$ new copies to node $1$ is upper bounded by $(K/N)^2$. Since the total number of failures on the time interval $[0,T]$ affecting node $1$ has a Poisson distribution with parameter $N{-}1$, one obtains that the probability that $(L_1^N(t))$ has a jump of size at least $2$ on $[0,T]$ is bounded by $K^2/N$ hence goes to $0$ as $N$ gets large.  The proposition is proved.
\end{proof}

\paragraph{Convergence to a Simple Jump Process}
\noindent Define
\[
{\cal P}^N[0,t]= \int_0^t \sum_{m=2}^N\int_{{\cal M}_U} z_i\left(L_m^N(s{-}),\overline{u}\right) \overline{{\cal N}}_{m}(\diff s,\diff \overline{u}),
\]
this is a counting process with jumps of size $1$. Define
\[
C_P^N(t)=\frac{1}{N-1}\sum_{m=2}^N \int_0^t  L_m^N(s)\,\diff s,
\]
then $(C_P^N(t))$ is the compensator of $({\cal P}^N[0,t])$ in the sense that it is a previsible process and that $$\left({\cal P}^N[0,t]{-}C_P^N(t)\right)$$ is a martingale. The proof is analogous to the proof of Proposition~\ref{propinc} in the Appendix. 
\begin{proposition}\label{PoisComp}
If the initial distribution of $(L_i^N(t))$  satisfies Condition~\eqref{IniCond} then,  for the convergence in distribution of processes,
  \[
  \lim_{N\to+\infty}(C_P^N(t))=(\beta t).
  \]
\end{proposition}
\begin{proof}
We first prove that the sequence $(C_P^N(t))$ is tight for the convergence in distribution with the topology of the uniform norm. 
By using that the sum of the $L_m$'s  equals $F_N$, for $0\leq s\leq t\leq T$,

\[
|C_P^N(t)-C_P^N(s)|=\frac{1}{N{-}1}\int_s^t \sum_{m=2}^N L_m^N(w)\,\diff w \leq \frac{F_N}{N{-}1}(t-s).
\]
Hence for any $\eta{>}0$ and $\eps{>}0$, there exists some $\delta>0$ such that, for all $N\geq 1$,
\[
\P\left(\sup_{\substack{0\leq s,t\leq T\\ |t-s|\leq \delta}}\left|C_P^N(t)-C_P^N(s)\right|\geq \eta \right)\leq \eps. 
\]
The sequence $(C_P^N(t))$ satisfies the criterion of the modulus of continuity, see Theorem 7.2 page~81 of Billingsley~\cite{Billingsley}. The property of tightness has been proved. Furthermore any limiting point  corresponds to a continuous process. 

The symmetry of the variables $(L_m(t))$ and the fact that their sum is $F_N$ give that
\[
\E(C_P^N(t))= \int_0^t \E\left(L_1^N(s)\right)\diff s=\frac{F_N}{N}t.
\]
Hence, the sequence $(\E(C_P^N(t)))$ is converging to $\beta t$. 

By using again the same arguments, one has 
\begin{multline*}
\E\left(C_P^N(t)^2\right)
{=} \E\left(\int_{[0,t]^2} \!\!\!\left[\frac{1}{N{-}1}\!\!\sum_{{1{\leq} m{\leq}N}}\!\!\!\!\!L_m^N(s)\right]\!\!\!\left[\frac{1}{N{-}1}\!\!\!\sum_{{1{\leq} m{\leq}N}}\!\!\!\!\!L_m^N(s')\right]\!\!\diff s \diff s'\!\!\right) \\
-2 \E\left(\int_{[0,t]^2} \frac{L_1(s)}{N{-}1}\left[\frac{1}{N{-}1}\sum_{{1{\leq} m{\leq}N}}L_m^N(s')\right]\,\diff s\diff s'\right)+\E\left(\int_0^t \frac{L_1(s)}{N{-}1}\,\diff s \right)^2,
\end{multline*}
hence
\begin{equation}\label{eqc2}
\E\left(C_P^N(t)^2\right){=}\!\left(\!\frac{F_N}{N{-}1}t\!\right)^2\!\!\!\!{-}2F_Nt\int_0^t\frac{\E(L_1^N(s))}{(N{-}1)^2}\,\diff+\E\left(\int_0^t \frac{L_1(s)}{N{-}1}\,\diff s
\right)^2.
\end{equation}

With Lemma~\ref{techprop} in Appendix, one obtains therefore that the second moment of $C_P^N(t)$ is converging to $(\beta t)^2$, hence 
\[
\lim_{N\to+\infty} \E\left(\left(C_P^N(t)^2-\beta t\right)^2\right)=0.
\]
One concludes that finite marginals of  the process $(C_P^N(t))$ converge to the corresponding marginals of $(\beta t)$. Consequently, $(\beta t)$ is the only limiting point of  $(C_P^N(t))$ for the convergence in distribution.  The tightness property gives therefore the desired convergence.  The proposition is proved. 
\end{proof}

\begin{theorem}\label{prop1}
If the initial distribution of $(L_i^N(t))$  satisfies Condition~\eqref{IniCond} and converges to some distribution $\pi_0$, then, for the convergence in distribution,
  \[
  \lim_{N\to+\infty} (L_1^N(t))= (X^R_{\beta}(t)),
  \]
  where $(X^R_{\beta}(t))$ is a jump process on $\N$ with initial distribution $\pi_0$  whose $Q$-matrix $Q{=}(q(x,y))$ is given by, for $x{\in}\N$,
  $q(x,x+1)=\beta$ and $q(x,0)=1$.
\end{theorem}
\begin{proof}
  By using Proposition~\ref{Q=L},  Proposition~\ref{PoisComp} and Theorem~5.1 of~\cite{Kasahara}, one concludes that the sequence of point processes $({\cal P}^N[0,t])$ is converging in distribution to a Poisson process ${\cal N}_\beta$ with rate $\beta$.

  Recall that $\overline{{\cal N}}_{1}(\diff t,{\cal M}_U )=(t_n^1)$, from SDE~\eqref{eqr}, one has
  \[
  \diff L_i^N(t)= {-}L_i^N(t{-}) \overline{{\cal N}}_{i}(\diff t,{\cal M}_U )+{\cal P}^N(\diff t),
  \]
  thus, for $t{>}0$,
  \[
  L_1^N(t)={\cal P}^N(t_n^1,t] \text{ if } t_n^1\leq t <t_{n+1}^1.
    \]
The convergence we have obtained shows that     $(L_1^N(t))$ is converging in distribution to $(\overline{L}_1(t))$ where
\[
\overline{L}_1(t)={\cal N}_{\beta}(t_n^1,t] \text{ if } t_n^1\leq t <t_{n+1}^1.
\]
This is the desired result. 
\end{proof}
This result explains the phenomenon observed in the simulations, Figure~\ref{fig:load_evol_199}, if a node has been alive for $t$ units of time, asymptotically it has received a Poissonnian number of files with rate $\beta t$, hence its average is growing linearly with $t$. 
\begin{proposition}\label{convprop}
The equilibrium distribution of $(L_1^N(t))$ is converging in distribution to  $\overline{X}^R_{\beta}$, a geometrically  distributed random variable with parameter $\beta/(1{+}\beta)$. 
\end{proposition}
\begin{proof}
  Denote by $\pi_{\beta}^N$ the invariant distribution of the process $(L_1^N(t))$. By symmetry, we know that $$  \E_{\pi_{\beta}^N}\left(L_1^N(0)\right){=}{F_N}/{N},$$ hence the sequence of probability distributions $(\pi_{\beta}^N)$ is tight. Let $\pi$ be some limiting point of this sequence for some subsequence $(N_k)$. If $f$ is some function on $\N$ with finite support, then the cycle formula for the invariant distribution of the ergodic Markov process $(L^N(t))$ gives the relation
  \[
\E_{\pi_{\beta}^N}(f(L_1))=\E_{\widehat{\pi}_{\beta}^N}\left(\int_0^{t_1^1} f(L_1^N(s))\,\diff s\right),
\]
where $\widehat{\pi}_{\beta}^N$ is the distribution of $(L^N(t))=(L_i^N(t))$ at the instants of jumps of breakdowns of node~$1$. In particular, $$\widehat{\pi}_{\beta}^N\left(\ell{=}(\ell_i){\in}\N^N,\, \ell_1{=}0\right)=1.$$

By Proposition~\ref{LemW2} in Appendix,  Theorem~\ref{prop1} is also true  when the initial distribution is $\widehat{\pi}_{\beta}^N$ hence, for the convergence in distribution,
\[
\lim_{N\to+\infty} \left(\int_0^{t_1^1} f(L_1^N(s))\,\diff s \right)=\left(\int_0^{t_1^1} f(X^R_\beta(s))\,\diff s\right),
\]
when the process $(X^R_\beta(t))$ has initial point $0$.  Consequently, by Lebesgue's Theorem,
\[
\E_{\pi}(f) =\lim_{N\to+\infty} \E_{\pi_{\beta}^N}(f(L_1))= \E_{0}\left(\int_0^{t_1^1} f(X^R_\beta(s))\,\diff s\right).
\]
The last term of this equation is precisely the invariant distribution of $(X^R_\beta(t))$, again with the cycle formula for ergodic Markov processes.  The  probability $\pi$ is necessarily the invariant distribution of $(X^R_{\beta}(t))$,   hence 
the sequence $(\pi_{\beta}^N)$ is converging to $\pi$.  It is easily checked that $\pi$ is a geometric distribution with parameter $\beta/(1{+}\beta)$.  
\end{proof}
By using the fact that
\[
\P\left(\overline{X}^R_{\beta}\geq n\right)=\left(\frac{\beta}{1+\beta}\right)^n,
\]
it is then easy to get the following result.
\begin{theorem}[Equilibrium at High Load]\label{th1}
The convergence in distribution
  \[
  \lim_{\beta\to+\infty} \frac{\overline{X}^R_{\beta}}{\beta}=E_1,
  \]
holds,  where $E_1$ is an exponential random variable with parameter~$1$. 
\end{theorem}
In particular the probability that, at equilibrium, the load of a given node is more than twice the average load is
\[
  \lim_{\beta\to+\infty}\P\left(\overline{X}^R_{\beta} \geq 2\beta\right)=\exp({-}2)\sim 0.135,
\]
which is consistent with the simulations, see Figure~\ref{fig:load_distrib}.

\subsection{The Power of Choice Algorithm}\label{SecPC}
Similarly as before, for $1\leq i\leq N$, $\overline{{\cal N}}_{i}=(t_n^i,(\overline{V}^{i}_n))$ denotes the marked Poisson point process defined as follows:

\noindent
\begin{itemize}
\item[---] $(t_n^i)$ is a Poisson process on $\R_+$ with rate $1$;
\item[---] $\overline{V}^{i}_n{=}(V^{i,n}_p){=}((V^{i,n}_{0,p},V^{i,n}_{1,p},B^{i,n}_p))$ where $(V^{i,n}_{0,p},V^{i,n}_{1,p})$ is an i.i.d. sequence with common distribution $(V_0,V_1)$ is uniform on the set of pairs of distinct elements of  $\{1,\ldots,N\}{\setminus}\{i\}$. Finally,   $(B^{i,n}_p)$ is i.i.d. Bernoulli sequence of random variables with parameter $1/2$.
\end{itemize}
The set of marks ${\cal M}_V$ is defined as
\[
{\cal M}_V{=}\left\{ \overline{v}{=}(v_{0,p},v_{1,p},b_p){\in}\{1,\ldots,N\}^2{\times}\{0,1\}{:} v_{0,p}{\not=}v_{1,p}\right\}
\]
For $1{\leq} i{\leq}N$ and $n{\geq} 1$,  $t_n^i$ is the instant of the $n$th breakdown of server $i$.  For $p\geq 1$,  $V^{i,n}_{0,p}$ and $V^{i,n}_{1,p}$ are the servers where the $p$th copy present on node $i$ may be allocated after this breakdown, depending on their respective loads of course. If the two loads are equal, the Bernoulli random  variable $B^{i,n}_p$ is then used. 
\vspace{2mm}

\paragraph{Equations of Evolution} 
\noindent For $1{\leq}i{\leq}N$ and $t{\geq}0$, $Q_i^N(t)$ is the number of copies on server~$i$ at time $t$ for this policy and $(Q^N(t)){=}(Q_i^N(t))$.
The dynamics of the power of choice algorithm is represented by the following stochastic differential equation, for $1{\leq} i{\leq} N$,
\begin{multline}\label{eqp}
\diff Q_i^N(t)= {-}Q_i^N(t{-}) \overline{{\cal N}}_{i}(\diff t,{\cal M}_V )\\+\sum_{m=1,m\neq i}^N\int_{\overline{v}\in {\cal M}_V} R^N_{mi}(Q^N(t{-}),\overline{v}) \,\overline{{\cal N}}_{m}(\diff t,\diff \overline{v})
\end{multline}
where $R^N_{mi}:\N^N{\times}{\cal M}_V\mapsto \N$ is the function, for $\ell=(\ell_k)$ and $\overline{v}{=}(v_{0,p},v_{1,p},b_p){\in}{\cal M}_V$, 
\[
R^N_{mi}(\ell, \overline{v}){=}\sum_{k=1}^{\ell_m} \ind{i\in\{v_{0,k},v_{1,k}\}} \left(\ind{\ell_i{<}\ell_{v_{0,k}}\!{\vee}\ell_{v_{1,k}}}{+}\ind{\ell_{v_{0,k}}{=}\ell_{v_{1,k}},i{=}v_{b_k,k}}\right).
\]
As it can be seen, when node $m$ breaks down while the network is in state $\ell$, $R^N_{mi}(\ell, \overline{v})$ is the number of copies sent to node $i$ by the power of choice policy if $\overline{v}$ is the corresponding mark associated to this instant.

Contrary to the random policy, the allocation depends on the state $(Q^N(t))$,  for this reason it is convenient to introduce 
the empirical distribution $\Lambda^N(t)$ as follows, if $f$ is some real-valued function on $\N$,
\[
\croc{\Lambda^N(t),f}=\int_\N f(\ell) \Lambda^N(t)(\diff \ell)=\frac{1}{N}\sum_{i=1}^N f\left(Q_i^N(t)\right).
\]
If $0{\leq} a{\leq} b$, $\croc{\Lambda^N(t),[a,b]}$ denotes $\Lambda^N(t)$ applied to the indication function of $[a,b]$.
In the same way as in the proof of Proposition~\ref{Q=L}, it can be proved that, with high probability and uniformly on any finite time interval, ${+}1$ is the unique value of  positive jumps of all processes.
By using  Equation~\eqref{eqr} and the definition of $\Lambda^N(t)$, one gets that, for a finite support function $f$,  with high probability, 
\begin{multline*}
  \diff\croc{\Lambda^N(t),f} = \diff M_f^N(t)+ \croc{\Lambda^N(t),f(0){-}f}\,\diff t
\\ +\sum_{\ell\in\N}\left[f\left(\ell{+}1\right){-}f\left(\ell\right)\right]
  \sum_{m} Q_m^N(t) \frac{1}{(N{-}1)(N{-}2)}\times \\ 
\bigg[ \sum_{\substack{j\not=j'\\j,j'\not=m}}\ind{Q_j^N(t)\geq \ell}\ind{Q_{j'}^N(t)\geq \ell}-\sum_{\substack{j\not=j'\\j,j'\not=m}}\ind{Q_j^N(t)\geq \ell+1}\ind{Q_{j'}^N(t)\geq \ell+1}\bigg]\,\diff t,
\end{multline*}
where $M_f(t)$ is a martingale. 
Note that the terms inside the brackets in the last equation is simply the number of pairs of nodes whose state is greater than $\ell$ and the state of at least one of them is $\ell$. 
By integrating, this  gives the relation
\begin{multline}\label{eqpow}
 \langle\Lambda^N(t),f\rangle{=} \langle\Lambda^N(0),f\rangle{+}M_f^N(t){+} \int_0^t \langle \Lambda^N(s), f(0){-}f\rangle\, \diff s\\
             {+} \beta \int_0^t \langle \Lambda^N(s), g_s\rangle \,\diff s+O(1/N),
\end{multline}
with
\[
g_s(\ell)=\left(f(\ell{+}1){-}f(\ell)\rule{0mm}{4mm}\right)\frac{\left[\Lambda^N(s)([\ell,{+}\infty))^2{-}\Lambda^N(s)([\ell {+}1,{+}\infty))^2\right]}{\Lambda^N(s)(\{\ell\}) }.
\]
\begin{proposition}[Mean-Field Convergence]\label{thP}\ 
  \begin{enumerate}
  \item  The distribution of $(Q_1^N(t))$ is converging in distribution to $({X}_{\beta}^P(t))$, a non-homogeneous Markov process whose $Q$-matrix $Q(t){=}(q(t)(x,y))$ is given by, for $x\in\N$,
    $q(t)(x,0)=1$ and 
  \[
 q(t)(x,x{+}1)  =\beta\frac{\P\big({X}_{\beta}^P(t)\geq x\big)^2-\P\big({X}_{\beta}^P(t)\geq x{+}1\big)^2}{\P\big({X}_{\beta}^P(t)=x\big)}
  \]
\item For the convergence in distribution, if $f$ has finite support, 
  \[
  \lim_{N\to+\infty} \left(\langle\Lambda^N(t),f\rangle\right)= \left(\E\left[f\left({X}_{\beta}^P(t)\right)\right]\right).
  \]
  \end{enumerate}
\end{proposition}
The proof which is quite technical is omitted to concentrate on the asymptotic behavior of the invariant distribution. It can be found in~\cite{RS2}. It is based on the proof of the convergence of the process $(\Lambda^N(t))$ by using Equation~\eqref{eqpow}. It is  similar in fact to the proof of an analogous result in the context of queuing systems, see Graham~\cite{Graham} for example.  The last reference contains also the existence and uniqueness result of such a non-homogeneous Markov process. 

\paragraph{The Invariant Distribution}
\noindent In this part, we study the asymptotic  behavior  of the invariant distribution of the load of a node at equilibrium.

\begin{proposition}\label{pibeta}
  The process  $({X}_{\beta}^P(t))$ of Proposition~\ref{thP} has a unique invariant distribution $\pi_\beta^P$ on $\N$,  which can be defined by induction as
  \[
\pi_\beta^P([x{+}1,{+}\infty)){=}\frac{-1{+}\sqrt{1{+}4\beta^2\pi_\beta^P([x,{+}\infty))^2}}{2\beta},\quad x\in\N,
    \]
    with $\pi_\beta^P([0,{+}\infty)){=}1$.
\end{proposition}
It should be noted that, due to the non-homogeneity of the Markov process, the uniqueness property is not clear in principle. 
\begin{proof}
  Let $\pi$ be an invariant probability of the process. If we start from this initial distribution, obviously the coefficients of the $Q$-matrix do not depend of  time, the invariant equations can be written as
  \[
  \begin{cases}  \pi(x)(1+q(x,x{+}1))=\pi(x{-}1)q(x{-}1,x), \, x>0,\\
    \pi(0)(1+q(0,1))=1.
  \end{cases}
  \]
  Define, for $x\geq 1$, $\xi(x){=}\pi(x{-}1)q(x{-}1,x)$, then $$\pi(x){=}\xi(x){-}\xi(x+1),$$ in particular $\P_{\pi}(X_{\beta}^P\geq x)=\xi(x)$, hence by definition of the $Q$-matrix
\begin{equation}\label{eqip}
  \xi(x+1)=\beta(\xi(x)^2-\xi(x+1)^2),
\end{equation}
hence, necessarily
\[
\xi(x{+}1){=}\frac{-1{+}\sqrt{1{+}4\beta^2\xi(x)^2}}{2\beta},
\]
with initial value $\xi(0){=}1$.  It is easily seen that the sequence $(\xi(x))$ is converging to $0$ so that $\pi$ is indeed a probability distribution on $\N$. The proposition is proved.
\end{proof}
\begin{proposition}\label{ConvInvPow}
 The invariant distribution of $(Q_1^N(t))$ is converging to the unique invariant distribution of $({X}_{\beta}^P(t))$. 
\end{proposition}
The proof is omitted, we refer to~\cite{RS2}. It shows that it is enough to analyze the invariant distribution $\pi_{\beta}^P$ of the limiting process we have just obtained. We can now state the main result of this section which explains the phenomenon observed in the simulations, see Figure~\ref{fig:load_distrib}.

\begin{theorem}[Equilibrium with High Load]\label{ThPow}
  If $\overline{X}_{\beta}^P$ is a random variable with distribution $\pi_{\beta}^P$, then, for the convergence in distribution,
  \[
  \lim_{\beta\to+\infty} \frac{\overline{X}_{\beta}^P}{\beta}=U,
  \]
  where $U$ is a uniformly distributed random variable on $[0,2]$. 
\end{theorem}
\begin{proof}
  In the proof of Proposition~\ref{pibeta}, we have seen that, by Equation~\eqref{eqip}, for $k\geq 0$,
  \begin{equation}\label{eqa2}
  \P\left(X_{\beta}^P{\geq}k{+}1\right){=}\beta\left(\P\left(X_{\beta}^P\geq k\right)^2{-}\P\left(X_{\beta}^P{\geq} k{+}1\right)^2\right),
  \end{equation}
  by summing these equations, one obtains
  \[
  \E\left(X_{\beta}^P {\wedge} x \right)=\beta\left(1-\P\left(X_{\beta}^P \geq x\right)^2\right),
  \]
where $a{\wedge}b=\min(a,b)$.  Hence, as expected, $\E(X_{\beta}^P){=}\beta$, and  therefore
  \begin{equation}\label{eqa1}
  \P\left(X_{\beta}^P \geq x\right)^2=\frac{1}{\beta}\left(\E\left(X_{\beta}^P\right)- \E\left(X_{\beta}^P {\wedge} x \right)\right)
  =\frac{1}{\beta}\E\left(\left(X_{\beta}^P- x \right)^+\right).
  \end{equation}
  By multiplying Equation~\eqref{eqa2} by $k{+}1$ and by summing up, one gets
\[
   \sum_{k=1}^{x} k \,\P\left(X_{\beta}^P{\geq} k\right) {=}\beta \sum_{k=0}^{x{-}1}  \P\left(X_{\beta}^P\geq k\right)^2{+}\beta\left(1{-}\P\left(X_{\beta}^P{\geq} x\right)^2\right).
\]
The right hand side of this relation is bounded by
\[
    \beta\left(\sum_{k=0}^{{+}\infty} \P\left(X_{\beta}^P\geq k\right)+1\right)=\beta(\beta+2)
\]
hence, by using Fubini's Theorem on the left hand side,
\[
\E\left(\left(X_{\beta}^P\right)^2\right)  {\leq} 2\beta(\beta{+}2),
\]
 so that 
\[
\sup_{\beta>0}\E\left(\left(\frac{X_{\beta}^P}{\beta}\right)^2\right){<}{+}\infty
\]
holds. In particular the family of random variables
\[
(Y_{\beta})\stackrel{\text{\rm def.}}{=}\left(\frac{X_{\beta}^P}{\beta}\right)
\]
is tight when $\beta$ goes to infinity. Let $Y$ be one of its limiting points, 
\[
  \P\left(Y_{\beta}  \geq x\right)^2=\E\left(\left(Y_{\beta}- \frac{\lceil x\beta\rceil}{\beta} \right)^+\right).
  \]
The uniform integrability property of $(Y_{\beta})$, consequence of the boundedness of the second moments, gives that $Y$ satisfies necessarily the relation
  \[
    \P\left(Y  \geq x\right)^2=\E\left(\left(Y- x\right)^+\right)=\int_x^{+\infty} \P(Y>s)\,\diff s.
    \]
    The function $f(x){=}\P(Y\geq x)$ is thus differentiable and satisfies the differential equation
    \[
    2f'(x)f(x){=}{-}f(x),
    \]
    for $x{\geq}0$, so that $f'(x){=}{-}1/2$ when $f(x){\not=}0$. One obtains the solution
    \[
    \P(Y\geq x)=\frac{(2-x)^+}{2},\quad x{\geq}0,
    \]
    with $a^+=\max(a,0)$, $Y$ is a uniformly distributed random variable on the interval $[0,2]$. The family of random variables $(Y_{\beta})$ has therefore  a unique limiting point when $\beta$ goes to infinity. One deduces the convergence in distribution. The theorem is proved.

\end{proof}

\section{Conclusion}\label{SecConc}
Our investigations through simulations and mathematical models have shown that
\begin{itemize}
\item[---] a simple, random placement strategy  may lead to heavily unbalanced situations;
\item[---] Classical load balancing techniques, like choosing the least loaded nodes are optimal from the point of view of placement. They have the drawback of requiring a detailed information on the state of the network, hence a significant cost in terms of complexity and bandwidth.
\item[---]  the \textit{power of two random choices}   policy has the advantage of having good performance with a limited cost in terms of storage space and of complexity.
\end{itemize}

\appendix

\section{Convergence Results}
The technical results of this section concern the random allocation scheme. The notations of the corresponding section are used. 
\begin{proposition}\label{propinc}
The previsible increasing process of the martingale $(M_1^N(t))$ is 
\begin{multline}\label{eq4}
\croc{M_1^N}(t)=\int_0^tL_1^N(s)^2\,\diff s \\+\sum_{m=2}^N \int_0^t\left[\frac{1}{(N{-}1)^2}L_m^N(s)^2+ \frac{N{-}2}{(N{-}1)^2} L_m^N(s)\right]\,\diff s.
\end{multline}
\end{proposition}
Concerning previsible increasing processes of martingales, see Section~VI-34 page~377 of Rogers and Williams~\cite{Rogers2}.
\begin{proof}
  The proof is not difficult, it is included for the sake of completeness for readers not familiar with the properties of martingales associated to marked Poisson point processes.   The previsible increasing process of the martingale 
\[
\left(\int_0^t L_1^N(s{-})\left[ \rule{0mm}{4mm}\overline{{\cal N}}_{1}(\diff t,{\cal M}_U )-\diff s\right]\right)
\text{ is } 
\left(\int_0^t L_1^N(s)^2\,\diff s\right),
\]
see Theorem~(28.1) page~50 of ~\cite{Rogers2}. By independence of the Poisson processes, it is enough to calculate the previsible increasing process of the martingale
\begin{multline*}
M_{1,m}^N(t)\stackrel{\text{\rm def.}}{=}\notag  \int_0^t \left[\int_{{\cal M}_U} z_1\left(L_m^N(s{-}),\overline{u}\right)\, \overline{{\cal N}}_{m}(\diff s,\diff \overline{u})- \frac{ L_m^N(s)}{N{-}1}\, \diff s\right]\\
=\sum_{t^m_n\leq t }\sum_{p=1}^{L_m^N(t_n^m{-})} \ind{U_p^{m,n}=1}-\int_0^t \frac{ L_m^N(s)}{N{-}1}\diff s,\label{eql}
\end{multline*}
for $2{\leq} m{\leq} N$. It is sufficient in fact to show that  the second moment of this martingale is such that
\[
\E\left(M_{1,m}^N(t)^2\right){=}\int_0^t\left[\frac{1}{(N{-}1)^2}\E\left(L_m^N(s)^2\right){+}\frac{N{-}2}{(N{-}1)^2} \E\left(L_m^N(s)\right)\right]\diff s,
\]
the property of independent increments of Poisson processes will then give the martingale property of $M_{1,m}^N(t)^2$ minus this term. By integrating with respect to the values of $(U_p^{m,n})$, one has
\[
\E\left(\left(\sum_{p=1}^{L_m^N(t_n^m{-})} \ind{U_p^{m,n}=1}\right)^2\right)=
\frac{N{-}2}{(N{-}1)^2}\E\left(L_m^N(t_n^m{-})\right)+ \frac{1}{(N{-}1)^2}\E\left(L_m^N(t_n^m{-})^2\right),
\]
which gives the relation
\begin{multline*}
  \E\left(\left(\sum_{t^m_n\leq t }\sum_{p=1}^{L_m^N(t_n^m{-})} \ind{U_p^{m,n}=1}\right)^2\right)
\\=\frac{1}{(N-1)^2}\E\left(\left(\int_0^t L_m^N(s{-})\,{\cal N}_m(\diff s)\right)^2\right)+\frac{N{-}2}{(N{-}1)^2}\E\left(\int_0^t L_m^N(s{-})\,{\cal N}_m(\diff s)\right).
\end{multline*}

In the same way, by integrating with respect to the values of $(U_p^{m,n})$, with the notation ${\cal N}_{m}(\diff s){=}\overline{{\cal N}}_{m}(\diff s,{\cal M}_U)$, 
\[
\E\left(\int_0^t \frac{ L_m^N(s)}{N{-}1}\diff s\sum_{t^m_n\leq t }\sum_{p=1}^{L_m^N(t_n^m{-})} \ind{U_p^{m,n}=1}\right)
=\E\left(\int_0^t \frac{ L_m^N(s)}{N{-}1}\diff s \int_0^t \frac{ L_m^N(s)}{N{-}1}\,{\cal N}_{m}(\diff s) \right).
\]
By using the last two relations one gets
\begin{multline*}
\E\left(M_{1,m}^N(t)^2\right)=\\
  \frac{1}{(N{-}1)^2} \E\left(\left(\int_0^t L_m^N(s)\left[{\cal N}_{m}(\diff s) -\diff s\right]\right)^2\right)
  +\frac{N{-}2}{(N{-}1)^2}\E\left(\int_0^t L_m^N(s{-})\,{\cal N}_m(\diff s)\right).
\end{multline*}
Since the martingale $({\cal N}_{m}([0,t]{-}t)$ associated to a Poisson process with rate $1$  has the increasing previsible process $(t)$, one gets
\[
  \E\left(M_{1,m}^N(t)^2\right)=
  \frac{1}{(N{-}1)^2} \int_0^t \E\left(L_m^N(s)^2\right)\diff s
  +\frac{N{-}2}{(N{-}1)^2}\int_0^t \E\left(L_m^N(s{-})\right)\diff s.
\]
The proposition is proved.
\end{proof}

\begin{lemma}\label{techprop} If the initial distribution of $(L_i^N(t))$  satisfies Condition~\eqref{IniCond} then, for any $T{>}0$,
   \begin{equation}\label{eq1}
\sup_{ N\geq 1}\,\sup_{0\leq t\leq T} \E\left((L_1^N(t))^2\right)<+\infty. 
  \end{equation}
and 
  \begin{equation}\label{eq2}
    \sup_{N\geq 1}\E\left(\sup_{0\leq s\leq T}L_1^N(s)\right) <+\infty.
  \end{equation}
\end{lemma}
\begin{proof}
With Relation~\eqref{eqr},  by writing the SDE satisfied by $(L_1^N(t)^2)$,
  \begin{multline*}
    L_1^N(t)^2=L_1^N(0)^2-\int_0^t L_1^N(s{-})^2\,{\cal N}_1(\diff s,{\cal M}_U)
\\+\sum_{m=2}^N \int_0^t \int_{{\cal M}_U}z_1\left(L_m^N(s{-}),\overline{u}\right)\times \left[2 L_1^N(s{-})+z_1\left(L_m^N(s{-}),\overline{u}\right)\right] \overline{{\cal N}}_{m}(\diff s,\diff \overline{u})
  \end{multline*}
by taking the expectation, one obtains
\begin{multline*}
  \E\left(L_1^N(t)^2\right)=  \E\left(L_1^N(0)^2\right)-\int_0^t   \E\left(L_1^N(s)^2\right)\,\diff s\\
  +\int_0^t \E\sum_{m=2}^N2\frac{L_m^N(s)}{N{-}1} L_1^N(s) \,\diff s+ \int_0^t\sum_{m=2}^N \E\left(\frac{L_m^N(s)(L_m^N(s){-1})}{(N{-}1)^2}+\frac{L_m^N(s)}{N{-}1}\right)\,\diff s
\end{multline*}
By using the fact that the $L_m^N(t)$'s have the same distribution and  their sum is $F_N$,  if
\[
f_N(t)=\E\left(L_1^N(t)^2\right),
\]
Equation~\eqref{eq3} gives that, for $0\leq t\leq T$,
\[
f_N(t)\leq (2T+1)\left\lceil \frac{F_N}{N}\right\rceil^2 +
\frac{1}{N{-}1}\int_0^t f_N(s)\,\diff s.
\]
If $p\in\N$ such that $F_N/N\leq p$ for all $N$, then, by Gronwall's Inequality, see Ethier and Kurtz~\cite{Ethier}~p.498,
\[
f_N(t)\leq p^2(1+2T)e^{ T/(N{-}1)}, \quad \forall N\geq 2. 
\]
Relation~\eqref{eq1} is proved.

\noindent Denote by $$S_m^N(t)=\sup\left(L_m^N(s):0\leq s\leq t\right)$$ then, by Equation~\eqref{eq3}, for $t{\leq}T$,
\[
S_1^N(t)\leq L_1^N(0)+\frac{1}{N{-}1}\sum_{m=2}^N \int_0^t S_m^N(s)\,\diff s+\sup_{0\leq s\leq T} |M_1^N(s)|.
\]
with the help of Doob's Inequality, see Theorem~(52.6) of Rogers and Williams~\cite{Rogers}, one gets
\[
\E\left(\sup_{0\leq s\leq T} M_1^N(s)^2\right)\leq 2\E\left(M_1^N(T)^2\right)=2\E\left(\croc{M_1^N}(T)\right)
\]
and this last quantity is bounded with respect to $N\geq 2$ by Relations~\eqref{eq4} and~\eqref{eq1}. Hence, by using the previous inequality, one can find a constant $K_0$ such that, for any $N{\geq} 2$ and $t{\leq}T$,
\[
\E\left(S_1^N(t)\right)\leq K_0 +\int_0^t \E\left(S_1^N(s)\right)\,\diff s,
\]
one concludes again with Gronwall's Inequality. The lemma is proved. 
\end{proof}

\begin{lemma}\label{LemW}
  If the initial condition $(L_j^N(0))$ is such that the variables $L_j^N(0)$, $j\geq 2$ are exchangeable and that
  \[
  \sup_{N\geq 1} \E\left(L_1^N(0)^2\right)+\E\left(L_2^N(0)^2\right)<+\infty
  \]
holds,   then, for all $T\geq 0$, 
  \[
  \sup_{N\geq 1}\sup_{0\leq t\leq T} \E\left(L_1^N(t)^2\right)+\E\left(L_2^N(t)^2\right)<+\infty,
  \]
 \end{lemma}
\begin{proof}
  The proof is similar to the proof of Lemma~\ref{techprop}. One has to introduce the functions
  \[
  f^1_N(t)=\E\left(L_1^N(t)^2\right) \text{ and } f^2_N(t)=\E\left(L_2^N(t)^2\right),
  \]
  by using an integral equation for $(L_1^N(t)^2)$ and $(L_1^N(t)^2)$  and the symmetry properties of the vector $(L_j^N(t), j{\geq 2})$, one obtains the relations
  \[
  \begin{cases}
\displaystyle    f_N^1(t)\leq  C_1+A_1\int_0^t f_N^1(s)\,\diff s+B_1\int_0^t f_N^2(s)\,\diff s,\\\ \\ 
\displaystyle    f_N^2(t)\leq  C_2+A_2\int_0^t f_N^1(s)\,\diff s+B_2\int_0^t f_N^2(s),
  \end{cases}
  \]
  for convenient positive constants $A_i$, $B_i$, $C_i$, $i=1$, $2$ independent of $N$. One uses Gronwall's Inequality for the first relation to get an upper bound on $f_N^1$,
  \[
  f_N^1(t)\leq  \left(C_1+B_1\int_0^t f_N^2(s)\,\diff s\right) e^{A_1 t}
  \]
  and Gronwall's Inequality is again used after plugging this relation in  the second inequality. 
\end{proof}
The next result is a technical extension of Proposition~\ref{PoisComp} used to prove Proposition~\ref{convprop}. 
\begin{proposition}\label{LemW2}
  If $\widehat{\pi}_{\beta}^N$ is the invariant distribution  of the state of the network at the instants of failures of node $1$, then, with the notations of Section~\ref{SecMath},  for the convergence in distribution, 
  \[
\lim_{N\to+\infty} \left(C_N^P(t)\right)=(\beta t)
\]
if the initial distribution of $(L^N(t))$ is $\widehat{\pi}_{\beta}^N$. 
\end{proposition}
\begin{proof}
  Let   $\widetilde{\pi}_{\beta}^N$ be the invariant distribution of  the process $(L_1^N(t))$ at the instants of failures on nodes, not only of node $1$. The sequence of states of the corresponding Markov chain is denoted as
  \[
  \left(\widetilde{L}^N_n\right)=\left(\widetilde{L}^N_{n,j}, 2\leq j\leq N\right)
  \]
  where $\widetilde{L}^N_{n,j}$, $2\leq j\leq N$ is the state of the nodes at the instant of the $n$th failure, i.e. the state of network reordered but with the failed node is excluded.
  If
  \[
  W_n=\left(\widetilde{L}^N_{n,2}\right)^2+\left(\widetilde{L}^N_{n,3}\right)^2+\cdots+\left(\widetilde{L}^N_{n,N}\right)^2,
  \]
  by invariance one has
  \[
  \E_{\widetilde{\pi}_{\beta}^N}(W_0)=  \E_{\widetilde{\pi}_{\beta}^N}(W_1),
  \]
  after some trite calculations, one obtains
  \[
  \E_{\widetilde{\pi}_{\beta}^N}\left(\left(\widetilde{L}^N_{0,2}\right)^2\right)=\frac{N{-}1}{N}\frac{F_N^2}{N^2}+\frac{F_N}{N}\frac{N{-}2}{N},
  \]
  hence
  \[
  \sup_{N\geq 2}   \E_{\widetilde{\pi}_{\beta}^N}\left(\left(\widetilde{L}^N_{0,2}\right)^2\right)<+\infty.
  \]
  The same property will hold when one considers only the instants of failures of node $1$ since, recall that $t_1^1$ is the first of these instants, 
  \[
  \widehat{\pi}_{\beta}^N\stackrel{\text{\rm dist.}}{=} \left(\widetilde{L}_i^N(t_1^1), i\geq2\right) \text{ if } \left(L_i^N(0)\right) \stackrel{\text{\rm dist.}}{=} {\widetilde{\pi}_{\beta}^N}. 
  \]
  By proceeding as in the proof of Lemma~\ref{techprop}, but  by stopping at time $t_1^1$ instead of a fixed time $t$, one obtains that
  \[
  \sup_{N\geq 2}\E_{  \widehat{\pi}_{\beta}^N}\left(L_2^N(0)^2\right)=  \sup_{N\geq 2}\E_{  \widetilde{\pi}_{\beta}^N}\left(\widetilde{L}_2^N(t_1^1)^2\right)<+\infty.
  \]
  Lemma~\ref{LemW} implies therefore that
  \[
\sup_{N\geq 2}\sup_{0\leq t\leq T}\E_{  \widehat{\pi}_{\beta}^N}\left(L_1^N(t)^2\right)+\E\left(L_2^N(t)^2\right)  <+\infty.
\]
One can now proceed as in the proof of Proposition~\ref{PoisComp} by noting that the crucial argument is the fact that the two last terms of the right hand side of Equation~\eqref{eqc2} vanish when $N$ gets large.
\end{proof}

\end{document}